%% file: focm_final_submitted.tex
\documentclass[11pt]{amsart}
\usepackage{graphicx,color,epsf}
\usepackage [cmtip,arrow]{xy}
\xyoption{all}

\usepackage {pb-diagram,pb-xy} 

\usepackage{amsmath}
\newtheorem{theorem}{Theorem}[section]
\newtheorem{lemma}[theorem]{Lemma}
\newtheorem{corollary}[theorem]{Corollary}
\newtheorem{proposition}[theorem]{Proposition}
\newtheorem{notationdefinition}[theorem]{Notation and definition}

\theoremstyle{definition}
\newtheorem{definition}[theorem]{Definition}
\newtheorem{example}[theorem]{Example}

\newtheorem{notation}[theorem]{Notation}

\theoremstyle{remark}
\newtheorem{remark}[theorem]{Remark}

\definecolor{DarkBlue}{rgb}{0,0.1,0.55}
\newcommand{\bfdef}[1]{{\bf \color{DarkBlue}  \emph{#1}}}

\numberwithin{equation}{section}

\newcommand {\hide}[1]{}

\newcommand {\junk}[1]{}

\newcommand {\R} {{\rm R}}

\newcommand {\K}     {\mbox{\rm K}}

\newcommand {\Sphere}{\mbox{${\bf S}$}}     
\newcommand {\Ball}{\mbox{${\bf B}$}}     

\newcommand {\Z}  {\mathbb{Z}}
 \newcommand {\N}         {\mathbb{N}}

\newcommand{\F}{\mathbb{F}}

\newcommand {\RR} {{\mathcal R}}

\newcommand {\V} {\mathbf{V}}

\def\addots{\mathinner{\mkern1mu
\raise1pt\vbox{\kern7pt\hbox{.}}
\mkern2mu\raise4pt\hbox{.}\mkern2mu
\raise7pt\hbox{.}\mkern1mu}}

\newcommand{\HH}  {\mbox{\rm H}}

\newcommand{\defeq}{\;{\stackrel{\text{\tiny def}}{=}}\;}

\newcommand{\dlim}{\mathop{\lim\limits_{\longrightarrow}}}

\newcommand{\x}{\mathbf{x}}
\newcommand{\X}{\mathbf{X}}

\newcommand{\y}{\mathbf{y}}
\newcommand{\Y}{\mathbf{Y}}

\newcommand{\z}{\mathbf{z}}
\newcommand{\ZB}{\mathbf{Z}}

\newcommand{\tb}{\mathbf{t}}
\newcommand{\TB}{\mathbf{T}}

\begin{document}
\title[]
{Polynomial hierarchy, Betti numbers and a 
real analogue of Toda's theorem 
}
\author{Saugata Basu}
\address{Department of Mathematics,
Purdue University, West Lafayette, IN 47906, U.S.A.}
\email{sbasu@math.purdue.edu}
\author{Thierry Zell}
\address{School of Mathematics and Computing Sciences,
Lenoir-Rhyne University, Hickory, NC  28603
}
\email{thierry.zell@lr.edu}

\thanks{Communicated by Peter Buergisser}

\thanks{The first author was supported in part by an NSF 
grant CCF-0634907. 
} 

\date{\textbf{\today}}


\begin{abstract}
Toda \cite{Toda} proved in 1989 that the (discrete) polynomial time hierarchy,
$\mathbf{PH}$, 
is contained in the class $\mathbf{P}^{\#\mathbf{P}}$,
namely the class of languages that can be 
decided by a Turing machine in polynomial time given access to an
oracle with the power to compute a function in the counting complexity
class $\#\mathbf{P}$. This result which illustrates the power of counting 
is considered to be a seminal result in computational complexity theory.
An analogous result in the complexity theory over the reals  
(in the sense of 
Blum-Shub-Smale real 
machines \cite{BSS89}) 
has been missing so far. In this paper we formulate and prove a real analogue
of Toda's theorem. Unlike Toda's proof in the discrete case, which relied
on sophisticated combinatorial arguments, our proof is topological in nature.
As a consequence of our techniques we 
are also able to relate the computational hardness
of two extremely well-studied problems in algorithmic semi-algebraic geometry 
-- namely the problem of deciding sentences in the first order theory of the 
reals 
with a constant number of quantifier alternations,
and that of computing Betti numbers of semi-algebraic sets.
We obtain a polynomial time reduction of the compact version of the
first problem to the second. This latter result might be of independent
interest to researchers in algorithmic semi-algebraic geometry.
\end{abstract}

\maketitle

\section{Introduction and Main Results}
\label{sec:intro}

\subsection{History and Background}
In this paper we study the relationship between the computational hardness
of two important classes of problems in algorithmic semi-algebraic geometry.
Algorithmic semi-algebraic geometry is concerned with designing efficient 
algorithms for deciding geometric as well as topological
properties of semi-algebraic sets. There is a large 
body of research in this area (see \cite{BPRbook2} for background).
If we consider the most important algorithmic problems studied in this area
(see for instance the survey article \cite{Basu_survey}),
it is possible to classify them into two broad sub-classes. The first class 
consists of the problem of quantifier elimination, 
and its special cases such as deciding a sentence in the first order 
theory of the reals, or deciding emptiness of semi-algebraic sets 
(also often called the existential theory of the reals). 
The existence of algorithms for solving these problems was
first proved by Tarski \cite{Tarski51}
and later research has aimed at designing algorithms with 
better complexities \cite{R92,Gri88,GV,BPR95}.
 
The second class of problems in algorithmic semi-algebraic geometry
that has been widely investigated consists of computing
topological invariants of semi-algebraic sets, 
such as counting the number of connected components,
computing the Euler-Poincar\'e characteristic, and more generally all the 
Betti numbers of semi-algebraic sets 
\cite{Canny93a,GV92,GR92,B99,BPR99,BPRbettione}. 
Note that the properties such as  
connectivity or the vanishing of some  Betti number of
a semi-algebraic set 
is not expressible in first-order logic, and thus the existence of 
algorithms for deciding such properties is not an immediate consequence 
of Tarski's result but usually requires some additional topological
ingredients such as semi-algebraic triangulations or Morse theory etc. 
Even though the most efficient algorithms for computing the
Betti numbers of a semi-algebraic set use algorithms for 
quantifier elimination in an essential way 
\cite{BPRbettione,Bas05-first}, 
the exact relationship between these two
classes of problems has not been clarified from the point of view of 
computational complexity and doing so is one of the motivations of this paper.

The primary 
motivation for this paper comes from 
classical (i.e., discrete) computational complexity theory. 
In classical complexity theory, there is a seminal result due to Toda
\cite{Toda}
linking the complexity of counting
with that of  deciding sentences with
a fixed number of quantifier alternations. 

More precisely, Toda's theorem gives the following inclusion 
(see \cite{Pap} for precise definitions of the complexity classes 
appearing in the theorem, see also Section \ref{subsec:counterparts} 
below for the corresponding classes over $\R$).

\begin{theorem}[Toda \cite{Toda}]
\label{the:toda}
\[
{\bf PH} \subset {\bf P}^{\#{\bf P}}.
\]
\end{theorem} 
In other words, any language in the (discrete) polynomial hierarchy can
be decided by a Turing machine in polynomial time, given access to an
oracle with the power to compute  a function in $\#\mathbf{P}$.

\begin{remark}
The proof of Theorem \ref{the:toda} in \cite{Toda} 
is quite non-trivial.
While it is obvious that the classes $\mathbf{P},\mathbf{NP},\mathbf{coNP}$
are contained in ${\bf P}^{\#{\bf P}}$, the proof for the higher
levels of the polynomial hierarchy is quite intricate and proceeds in 
two steps: first proving that the $\mathbf{PH} \subset 
\mathbf{BP} \cdot\oplus\cdot \mathbf{P}$ (using previous  results of
Sch{\"o}ning \cite{Schoning},  and Valiant and Vazirani \cite{VV}), 
and then showing that
$\mathbf{BP} \cdot\oplus\cdot \mathbf{P} \subset {\bf P}^{\#{\bf P}}$.
Aside from the obvious question about what should be a proper analogue of
the complexity class $\# \mathbf{P}$ over the reals, 
because of the presence of the intermediate complexity class in the proof,
there seems to be no direct way of extending such a proof 
to real complexity classes in the sense of the 
Blum-Shub-Smale model of computation \cite{BSS89,Shub-Smale96}. 
The proof of the main theorem (Theorem \ref{the:main}) of this paper,
which can be seen as a real analogue of Theorem \ref{the:toda}, proceeds
along completely different lines and is mainly topological in nature.
\end{remark}

In the late eighties Blum, Shub and Smale \cite{BSS89,Shub-Smale96}
introduced the notion
of Turing 
machines over more general fields, thereby generalizing
the classical problems of computational complexity theory such as
$\mathbf{P}$ vs. $\mathbf{NP}$ to corresponding problems over arbitrary fields
(such as the real, complex, $p$-adic numbers etc.). If one considers 
languages accepted by a  Blum-Shub-Smale machine 
over a finite field one recovers the classical notions of discrete complexity
theory. Over the last two decades there has been a lot of research activity 
towards proving real as well as complex analogues of well known theorems
in discrete complexity theory. The first steps in this direction were taken by
the authors 
Blum, Shub, and Smale (henceforth B-S-S) 
themselves, when they proved the 
$\mathbf{NP}_{\R}$-completeness of the problem of deciding whether
a  real polynomial equation in many variables of degree at most four 
has a real solution
(this is the real analogue of Cook-Levin's theorem that the satisfiability
problem is $\mathbf{NP}$-complete in the discrete case), and subsequently 
through the work of several  researchers 
(Koiran, B{\"u}rgisser, Cucker, Meer to name 
a few) a well-established complexity theory over the reals as well as 
complex numbers have been built up, which mirrors closely the discrete case.

From the point of view of computational complexity theory of 
real B-S-S machines the classes 
${\bf PH}$ and ${\#{\bf P}}$ appearing in the two sides of the 
inclusion in Theorem \ref{the:toda} can be 
identified with the two broad classes of problems in algorithmic
semi-algebraic geometry discussed previously, viz. the polynomial
hierarchy with the problem of deciding sentences with a fixed number
of quantifier alternations, and the class $\#{\bf P}$ with the problem
of computing certain topological invariants of semi-algebraic sets,
namely their Betti numbers which generalize the notion of cardinality
for finite sets.  (This naive intuition is made more precise in Section~\ref{sec:sharpP}.)
It is thus quite natural to seek a real analogue of Toda's theorem.
Indeed, there has been a large body of recent research on obtaining 
appropriate real (as well as complex) analogues of results in discrete
complexity theory, especially those related to counting complexity classes
(see \cite{Meer00,BC2,Burgisser-Cucker05,Burgisser-Cucker06}).

In order to formulate
such a result it is first necessary to define precisely 
real counter-parts of the discrete 
polynomial time hierarchy ${\bf PH}$ and
the discrete complexity class $\# {\bf P}$,
and this is what we do next.

\subsection{Real counter-parts of ${\bf PH}$ and $\# {\bf P}$}
\label{subsec:counterparts}
For the rest of the paper $\R$ will denote a real closed field (there
is no essential loss in assuming that $\R = {\mathbb R}$). 
By a real 
machine we will mean a  machine in the sense of 
Blum-Shub-Smale \cite{BSS89}
over the ground field $\R$.

\medskip
\noindent
{\em Notational convention.}
Since in what follows we will be forced to deal with multiple blocks of
variables in our formulas, we follow a notational convention by
which we denote blocks of variables by bold letters with superscripts 
(e.g. $\X^i$ denotes the $i$-th block), and we use non-bold letters
with subscripts to denote single variables (e.g. $X^i_j$ denotes the
$j$-th variable in the $i$-th block).
We use $\x^i$ to denote a specific value of the block of variables $\X^i$.

\subsubsection{Real analogue of ${\bf PH}$}
We 
introduce
the polynomial hierarchy for the reals.
It  mirrors the discrete case very closely (see \cite{Stockmeyer}).

\begin{definition}[The class $\mathbf{P}_\R$]
Let $k(n)$ be any polynomial in $n$.
A sequence of semi-algebraic sets
$(T_n \subset \R^{k(n)})_{n > 0}$ 
is said to belong to the class $\mathbf{P}_\R$
if there exists a 
machine $M$ over $\R$ 
(see~\cite{BSS89} or~\cite[\S3.2]{BCSS98}),
such that for all $\x \in \R^{k(n)}$, the machine $M$ 
tests membership of $\x$ in $T_n$ 
in time bounded by a polynomial in $n$.
\end{definition}

\begin{definition}\label{df:sigma}
Let $k(n),k_1(n),\ldots,k_\omega(n)$ be polynomials in $n$.
A sequence of semi-algebraic sets
$(S_n \subset \R^{k(n)})_{n > 0}$ is said to be in 
the complexity class ${\bf \Sigma}_{\R,\omega}$, 
if for each $n > 0$ the semi-algebraic set
$S_n$ is described by a 
first order formula
\begin{equation}\label{eq:alternation}
(Q_1 \Y^{1} )  \cdots (Q_\omega \Y^{\omega} )
\phi_n(X_1,\ldots,X_{k(n)},\Y^1,\ldots,\Y^\omega),
\end{equation}
with $\phi_n$ a quantifier free formula in the first order theory of the reals,
and for each $i, 1 \leq i \leq \omega$,
$\Y^i = (Y^i_1,\ldots,Y^i_{k_i(n)})$ is a block of $k_i(n)$ variables,
$Q_i \in \{\exists,\forall\}$, with $Q_j \neq Q_{j+1}, 1 \leq j < \omega$,
$Q_1 = \exists$,
and 
the sequence of semi-algebraic sets  $(T_n \subset \R^{k(n)+ k_1(n) + \cdots + k_\omega(n)})_{n >0}$ 
defined by the quantifier-free formulas $(\phi_n)_{n>0}$ 
belongs to the class $\mathbf{P}_\R$.
\end{definition}

Similarly, the complexity class ${\bf \Pi}_{\R,\omega}$
is defined as in Definition~\ref{df:sigma}, with the exception 
that the  alternating quantifiers in~\eqref{eq:alternation} start with $Q_1=\forall$.
Since, adding an additional block of quantifiers on the outside
(with new variables) does not change the set defined by a quantified formula
we have the following inclusions:
$$ {\bf \Sigma}_{\R,\omega}\subset {\bf \Pi}_{\R,\omega+1}, 
\text{\ and\ }
{\bf \Pi}_{\R,\omega}\subset{\bf \Sigma}_{\R,\omega+1}.
$$

Note that by the above definition the class 
${\bf \Sigma}_{\R,0} = {\bf \Pi}_{\R,0}$ 
is the 
familiar class ${\bf P}_{\R}$, 
the class ${\bf \Sigma}_{\R,1} = {\bf NP}_{\R}$ and the
class ${\bf \Pi}_{\R,1} = {\bf \mbox{co-}NP}_{\R}$.

\begin{definition}[Real polynomial hierarchy] 
The real polynomial time hierarchy is defined to be the union
\[
{\bf PH}_{\R} \defeq \bigcup_{\omega \geq 0} 
({\bf \Sigma}_{\R,\omega} \cup {\bf \Pi}_{\R,\omega}) = 
\bigcup_{\omega \geq 0} {\bf \Sigma}_{\R,\omega}  = 
\bigcup_{\omega \geq 0} {\bf \Pi}_{\R,\omega}.
\]
\end{definition}

For technical reasons (see Remark~\ref{rem:restrictiontocompacts})
we need to restrict to compact semi-algebraic sets,
and for this purpose, we will now define a compact analogue of
${\bf PH}_{\R}$ that we will denote ${\bf PH}_{\R}^c$.

\begin{definition}
We call $K \subset \R^n$ a \bfdef{semi-algebraic compact} if it is
a closed and bounded semi-algebraic set. (Note that if
$\R\neq\mathbb{R}$, $K$ is not necessarily compact in the order
topology.)
\end{definition}

\begin{notation}
We denote by $\Ball^k(0,r)$ the closed ball in $\R^k$ of radius $r$ centered at the
origin. We will denote by $\Ball^k$ the closed unit ball $\Ball^k(0,1)$.
Similarly, we denote by $\Sphere^k(0,r)$ the sphere in $\R^{k+1}$ 
of radius $r$ centered at the origin, and 
by $\Sphere^k$ the unit sphere $\Sphere^k(0,1)$.
\end{notation}

We now define our compact analogue ${\bf PH}_{\R}^c$ of the real polynomial hierarchy
${\bf PH}_{\R}$. Unlike in the non-compact case, we will assume all variables
vary over certain compact semi-algebraic sets 
(namely spheres of varying dimensions).

\begin{definition}[Compact real polynomial hierarchy]
\label{def:compactpolynomialhierarchy}
Let \[
k(n),k_1(n),\ldots,k_\omega(n)\]
be polynomials in $n$.
A sequence of semi-algebraic sets
$(S_n \subset \Sphere^{k(n)})_{n > 0}$ 
is in the complexity class ${\bf \Sigma}_{\R,\omega}^c$, 
if for each $n > 0$ the semi-algebraic set
$S_n$ is described by a first order formula
\[
 (Q_1 \Y^{1} \in \Sphere^{k_1(n)})  \cdots (Q_\omega \Y^{\omega} \in 
\Sphere^{k_\omega(n)} )
\phi_n(X_0,\ldots,X_{k(n)},\Y^1,\ldots,\Y^\omega),
\]
with $\phi_n$ a quantifier-free first order formula defining a 
{\em closed} semi-algebraic subset of 
$\Sphere^{k_1(n)}\times\cdots\times
\Sphere^{k_\omega(n)}\times \Sphere^{k(n)}$
and for each $i, 1 \leq i \leq \omega$,
$\Y^i = (Y^i_0,\ldots,Y^i_{k_i})$ is a block of $k_i(n)+1$ variables,
$Q_i \in \{\exists,\forall\}$, with $Q_j \neq Q_{j+1}, 1 \leq j < \omega$,
$Q_1 = \exists$, and
the sequence of semi-algebraic sets
$(T_n \subset \Sphere^{k_1(n)}\times\cdots\times 
\Sphere^{k_\omega(n)}\times\Sphere^{k(n)})_{n > 0}$
defined by the formulas $(\phi_n)_{n >0}$ belongs to the class
$\mathbf{P}_\R$.

\begin{example}
\label{eg:compact}
The following is an example of a language in $\mathbf{\Sigma}_{\R,1}^c$ (i.e., 
of the compact version of $\mathbf{NP}_\R$).

Let $k(n) = \binom{n+4}{4}-1$ and identify $\R^{k(n)+1}$ with the
space of \emph{homogeneous} polynomials in $\R[X_0,\ldots,X_n]$ of degree
$4$. Let $S_n \subset \Sphere^{k(n)} \subset \R^{k(n)+1} $ be defined by 
\[
S_n = \{P \in  \Sphere^{k(n)} \;\mid\; \exists \x 
=(x_0:\cdots:x_n) \in \mathbb{P}_\R^n \mbox{ with }
P(\x) = 0 \};
\]
in other words  $S_n$ is the set of (normalized) real forms of degree $4$ 
which have a zero in the real projective space $\mathbb{P}^n_\R$.
Then 
\[
(S_n \subset \Sphere^{k(n)})_{n > 0} \in \mathbf{\Sigma}_{\R,1}^c,
\]
since it is easy to see that $S_n$ also admits the description:
\[
S_n = \{P \in  \Sphere^{k(n)} \;\mid\; \exists \x \in \Sphere^n \mbox{ with }
P(\x) = 0 \}.
\]

Note that it is \emph{not known} if $(S_n \subset \Sphere^{k(n)})_{n > 0}$ 
is $\mathbf{NP}_\R$-complete 
(see Remark \ref{rem:compact}),
while the non-compact version of this language, 
i.e., the language consisting of (possibly non-homogeneous) 
polynomials of degree at most four having a 
zero in $\mathbb{A}_\R^n$ (instead of $\mathbb{P}^n_\R$),  
has been shown to be $\mathbf{NP}_\R$-complete \cite{BCSS98}.
\end{example}

We define analogously the class ${\bf \Pi}_{\R,\omega}^c$, and finally 
define the \bfdef{compact real polynomial time hierarchy} 
to be the union
\[
{\bf PH}_{\R}^c \defeq 
\bigcup_{\omega \geq 0} ({\bf \Sigma}_{\R,\omega}^c \cup 
{\bf \Pi}_{\R,\omega}^c) = \bigcup_{\omega \geq 0} {\bf \Sigma}_{\R,\omega}^c =
\bigcup_{\omega \geq 0} {\bf \Pi}_{\R,\omega}^c.
\]
\end{definition}

Notice that the semi-algebraic sets belonging to any language in 
${\bf PH}_{\R}^c$ are all semi-algebraic compact 
(in fact closed semi-algebraic  subsets of
spheres). Also, note the inclusion
\[
{\bf PH}_{\R}^c \subset {\bf PH}_{\R}.
\]

\begin{remark}
\label{rem:compact}
Even though the restriction to compact semi-algebraic sets might appear
to be only a technicality at first glance, 
this is actually an important restriction.
For instance, it is a long-standing  open question in real complexity theory 
whether there exists an  ${\bf NP}_{\R}$-complete 
problem which belongs to the class ${\bf \Sigma}_{\R,1}^c$ (the compact version
of the class ${\bf NP}_{\R}$, 
see Example \ref{eg:compact}). 
(This distinction between 
compact and non-compact versions of complexity classes
does not arise in discrete complexity theory for obvious reasons.) 
It is an interesting question whether
the main theorem of this paper can be extended to 
the full class ${\bf PH}_{\R}$. For technical reasons which will become clear
later in the paper (Remark \ref{rem:restrictiontocompacts})
we are unable to achieve this presently.   
\end{remark}

\begin{remark}
The topological methods used in this paper only require the sets to be compact. Using spheres to achieve this compact situation is 
a natural choice in the context of real algebraic geometry, 
since the inclusion of the space $\R^n$ 
into its one-point compactification $\Sphere^n$ is a 
continuous semi-algebraic map that sends semi-algebraic subsets of $\R^n$ 
to their own one-point compactifications (see \cite[Definition~2.5.11]{BCR}).
\end{remark}

\subsubsection{Real Analogue of $\#{\bf P}$}
\label{sec:sharpP}
Before defining the real analogue of the class $\#{\bf P}$, let us recall
its definition in the discrete case which is well known.

 \begin{figure}[hbt]
         \centerline{
           \scalebox{0.5}{
             \input{sharpP.pstex_t}
             }
           }
         \caption{The fibers $L_{m+n,\x}$ of the language $L$}
         \label{fig:sharpP}
 \end{figure}
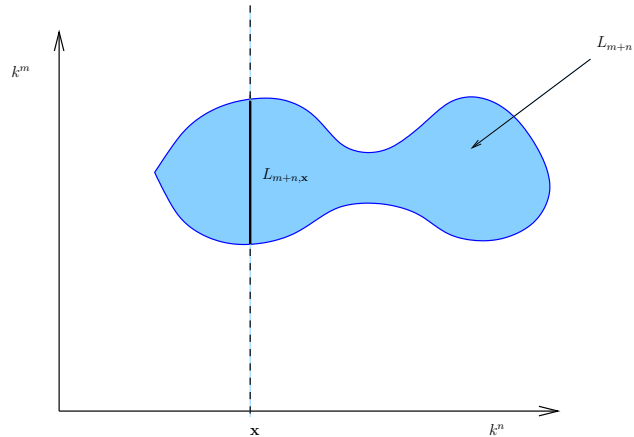

\begin{definition}
\label{def:sharpPdiscrete}
Let $k = \Z/2 \Z$. 
We say that a sequence of functions 
\[
(f_n:k^n \rightarrow \N)_{n > 0}
\]
is in the class $\#{\bf P}$ 
if there exists 
a language
$$
\displaylines{
L = \bigcup_{n > 0} L_n, \;L_n \subset k^n, \mbox{ with } L \in \mathbf{P}
}
$$
as well as a polynomial $m(n)$,
such that
\[
f_{n}(\x) = \mathrm{card} (L_{m+n,\x}) 
\]
for each $\x \in k^n$, where $L_{m+n,\x} = L_{m+n} \cap \pi^{-1}(\x)$ 
and $\pi:k^{m+n} \rightarrow k^n$ is the projection along the first 
$m$ co-ordinates. 
\end{definition}

In other words, $f_n$ \emph{counts} the number of points in
the fibers, $L_{m+n,\x}$,  of a language $L$ belonging to the 
(discrete) complexity class $\mathbf{P}$ (see Figure \ref{fig:sharpP}). 
(The geometric language used above might look unnecessary  
but it is very helpful towards obtaining the right 
analogue in the real case.)

In order to define real analogues of counting complexity classes of discrete
complexity theory, it is necessary to identify the proper notion of 
``counting'' in the context of semi-algebraic geometry. Counting complexity 
classes over the reals have been defined previously by Meer 
\cite{Meer00}, and studied extensively by other authors 
\cite{Burgisser-Cucker06}. 
These authors used a straightforward generalization to semi-algebraic 
sets of counting in the case of finite sets -- 
namely the counting function took the value of the 
cardinality of a semi-algebraic set if it happened to be finite, and 
$\infty$ otherwise. This is in our view not a fully satisfactory 
generalization since 
the count gives no information when the  semi-algebraic set is infinite, and 
most interesting semi-algebraic sets have infinite cardinality. 
Moreover, no real analogue of Toda's theorem has been proved using this 
definition of counting.

If one thinks of ``counting'' a semi-algebraic set $S\subset \R^k$ 
as computing a certain  discrete invariant, then a natural well-studied 
discrete topological invariant of $S$ is its   
Euler-Poincar\'e characteristic. 
For a closed and bounded semi-algebraic set $S$ 
the Euler-Poincar\'e characteristic, $\chi(S)$, is the alternating sum of the
Betti numbers of $S$, and it is possible to extend this definition to the class
of all semi-algebraic sets by additivity. 
This generalized Euler-Poincar\'e characteristic 
gives an isomorphism from the \bfdef{Grothendieck ring
of semi-algebraic sets} to $\Z$ 
(see \cite[Proposition 1.2.1]{Cluckers-Loeser08})), and 
thus corresponds to a mathematically natural notion of counting
semi-algebraic sets.

However, the Euler-Poincar\'e characteristic fails to 
distinguish between empty and non-empty semi-algebraic sets, since
a non-empty semi-algebraic set (e.g. an odd dimensional sphere) can
have have vanishing Euler-Poincar\'e characteristic. This seems to
immediately rule out using the Euler-Poincar\'e characteristic as a 
substitute for the  counting function. We make up for this deficiency
by replacing the Euler-Poincar\'e characteristic by the 
Poincar\'e polynomial $P_S(T)$ of the set $S$.
We now recall the relevant definitions.

\begin{notation}
For any semi-algebraic set $S\subset \R^k$ we denote by $b_i(S)$ the $i$-th
Betti number (that is the rank of the singular 
homology group $\HH_i(S) = \HH_i(S,\Z)$) of $S$.

We also let $P_S \in \Z[T]$ denote the \bfdef{Poincar\'e polynomial} of $S$, namely

\begin{equation}
\label{def:Poincarepolynomial}
P_S(T)\; \defeq\; \sum_{i \geq 0}  b_i(S)\; T^i.
\end{equation}
\end{notation}
Notice that for $S \subset \R^k$, $\deg(P_S) \leq k-1$. 
Also, it is easy to see that the Poincar\'e polynomial, $P_S(T)$,  carries
more complete information about $S$ than its Euler-Poincar\'e characteristic. 
Indeed, the number of semi-algebraically connected components, 
$b_0(S)$, of  $S$ is obtained by setting $T$ to $0$, 
and in case $S$ is closed and bounded 
we also recover $\chi(S)$ by setting $T$ to $-1$
in $P_S(T)$. Since $b_0(S) > 0$ if and only if $S$ is 
non-empty, $P_S$, unlike $\chi(S)$, can distinguish between
empty and non-empty semi-algebraic sets.
In particular, in case $S$ is a finite set of points, 
$P_S$ also contains the information regarding the cardinality of $S$ which 
in this case equals  $b_0(S) =P_S(0)$.

\begin{remark}
\label{rem:zeta}
The connection between counting points of varieties and their Betti numbers is 
more direct over fields of positive characteristic via the zeta function.
The zeta function of a variety defined over $\F_p$ is the exponential 
generating function of the sequence whose $n$-th term is the number of
points in the variety over $\F_{p^n}$. The zeta function of such a variety
turns out to be a rational function in one variable 
(a deep theorem of algebraic geometry first conjectured by Andre Weil
\cite{Weil}
and proved by Dwork \cite{Dwork} and  Deligne \cite{Deligne1, Deligne2}), 
and its numerator and denominator are products of 
polynomials whose degrees are the  Betti numbers of the variety 
with respect to a certain ($\ell$-adic) co-homology theory. The point of this 
remark is that the problems  of  ``counting'' varieties and computing their 
Betti numbers, are connected at a deeper level, and thus our 
choice of definition for a real analogue 
of $\#{\bf P}$ is not altogether ad hoc.
\end{remark}

The above considerations motivate us to depart 
from the definition of $\# {\bf P}_{\R}$ considered
previously in \cite{Meer00,Burgisser-Cucker06}. 
We denote our class  $\# {\bf P}_{\R}^{\dagger}$ to avoid any 
possible confusion with these authors' work.

\begin{definition}[The class $\#{\bf P}_{\R}^{\dagger}$] 
\label{def:sharp}
We say a sequence of functions 
\[
(f_n:\R^n \rightarrow \Z[T])_{n > 0}
\]
is in the class $\#{\bf P}_{\R}^{\dagger}$,
if there exists a language
\[
(S_n \subset \R^n)_{n > 0} \in \mathbf{P}_\R,
\]
as well as a polynomial $m(n)$,
such that
\[
f_{n}(\x) = P_{S_{m+n,\x}} 
\]
for each $\x \in \R^n$, where $S_{m+n,\x} = S_{m+n} \cap \pi^{-1}(\x)$ 
and $\pi:\R^{m+n} \rightarrow \R^n$ is the projection along the first 
$m$ co-ordinates. 
\end{definition}

\begin{remark}
Notice the formal similarity between Definitions \ref{def:sharp} and
\ref{def:sharpPdiscrete}, namely that in both cases the functions $f_n$   
``counts'' the fibers above $\x$, but the notion of counting is different
in each case.
\end{remark}

\begin{remark}
We make a few remarks about the class $\#{\bf P}_{\R}^{\dagger}$ 
defined above. First of all notice that
the  class $\#{\bf P}_{\R}^{\dagger}$ is quite robust.
For instance, given two sequences 
$(f_n)_{n > 0}, (g_n)_{n > 0} \in \#{\bf P}_{\R}^{\dagger}$ it follows
(by taking disjoint union of the corresponding semi-algebraic sets) that 
$(f_n+ g_n)_{n > 0} \in \#{\bf P}_{\R}^{\dagger}$, and also 
$(f_n g_n)_{n > 0} \in \#{\bf P}_{\R}^{\dagger}$ 
(by taking Cartesian product of the corresponding semi-algebraic sets
and using the multiplicative
property of the Poincar\'e polynomials, which itself is a consequence
of the Kunneth formula in homology theory.)

Secondly, note that while it is still an open question 
whether the Poincar\'e  polynomials of semi-algebraic sets 
can be computed in single exponential
time, there exists single exponential time algorithms to calculate its
first few (i.e., a constant number of) coefficients 
(or equivalently the first few Betti numbers of a given semi-algebraic)
\cite{Bas05-first}. 
Also, it is known that the closely related invariant, namely the
Euler-Poincar\'e characteristic can be computed in single exponential 
time \cite{B99,BPRbook2}. 
The reader is referred to the survey article \cite{Basu_survey} 
for a more  detailed account of the state-of-the-art ragarding these 
algorithmic problems.
\end{remark}

\subsection{Statements of the main theorems}
We can now state the main result of this paper.
\begin{theorem}[Real analogue of Toda's theorem]
\label{the:main}
\[
{\bf PH}^c_{\R} \subset {\bf P}_{\R}^{\#{\bf P}_{\R}^{\dagger}}.
\]
\end{theorem}

Recall that the class ${\bf P}_{\R}^{\#{\bf P}_{\R}^{\dagger}}$ 
consists of sequences of semi-algebraic sets, such that
there exists a real machine which can check membership
in the sets in the sequence in polynomial time, given access to an
oracle with the power to compute  a function in 
${\#{\bf P}_{\R}^{\dagger}}$.

\begin{remark}
We leave it as an open problem to prove Theorem \ref{the:main} with
${\bf PH}_{\R}$ instead of ${\bf PH}^c_{\R}$ on the left hand side. One 
possible
approach would be to use the recent results of Gabrielov and Vorobjov 
\cite{GV07} 
on replacing arbitrary semi-algebraic sets by compact semi-algebraic sets 
in  the same homotopy equivalence class using infinitesimal deformations. 
However, for  such a construction to be useful in our context, one would 
need to effectively
(i.e., in polynomial time) replace the infinitesimals used in the
construction by small enough positive elements of $\R$, and at present we are
unable to achieve this.
\end{remark}

As a consequence of our method, we obtain a reduction
(Theorem~\ref{the:main2})
that might be of independent interest.
We first define the following two problems:

\begin{definition}(Compact general decision problem with at most 
$\omega$ quantifier alternations (${\bf GDP_\omega^c}$))

\begin{itemize}
\item[Input.] 
A sentence $\Phi$ in the first order theory of $\R$
\[    (Q_1 \X^{1} \in \Sphere^{k_1})  \cdots 
(Q_\omega \X^{\omega} \in \Sphere^{k_\omega})
\phi(\X^1,\ldots,\X^\omega),
\]
where for each $i, 1 \leq i \leq \omega$,
$\X^i = (X^i_0,\ldots,X^i_{k_i})$ is a block of $k_i+1$ variables,
$Q_i \in \{\exists,\forall\}$, with $Q_j \neq Q_{j+1}, 1 \leq j < \omega$,
and
$\phi$ is a quantifier-free formula defining a {\em closed}
semi-algebraic subset $S$ of $\Sphere^{k_1}\times\cdots\times\Sphere^{k_\omega}$. 
\item[Output.] True or False depending on whether $\Phi$ is true or false
in the first order theory of $\R$.
\end{itemize}
\end{definition}

\begin{definition}(Computing the Poincar\'e polynomial of 
semi-algebraic sets (\textbf{Poincar\'e}))

\begin{itemize}
\item[Input.] A quantifier-free formula defining a 
semi-algebraic set $S\subset \R^k$.
\item[Output.] 
The Poincar\'e polynomial $P_S(T)$.
\end{itemize}
\end{definition}

\begin{theorem}
\label{the:main2}
For every  $\omega > 0$, there is a deterministic 
polynomial time reduction in the 
Blum-Shub-Smale model of 
$\bf{GDP_\omega^c}$ to \upshape \textbf{Poincar\'e}.
\end{theorem}

\subsection{Summary of the main ideas}
Our main tool is a topological construction described in 
Section~\ref{sec:saconstructions}, which,
given a semi-algebraic set 
$S \subset \Sphere^m\times \Sphere^n$ which is open or closed and defined 
by some \emph{quantified} formula $\Phi(\X,\Y)$, and given an integer
$p \geq 1$, constructs
{\em efficiently}  a semi-algebraic set, $D_\Y^p(S)$, such that 
\begin{equation}\label{eq:D}
b_i(\pi_\Y(S)) = b_i(D_\Y^p(S)), \quad 0 \leq i < p
\end{equation}
(where $\pi_\Y$ is the projection along the $\Y$ coordinates.)
Moreover such that the description of $D_\Y^p(S)$ requires fewer quantifier blocks than the number of quantifier blocks that appear 
in the formula 
$$ \exists \Y \Phi(\X,\Y);$$
that describes $\pi_\Y(S)$.
(For technical reasons 
we need two different constructions depending on whether $S$ is an
open or a closed semi-algebraic set, but we prefer to 
ignore this point in this rough outline.)
An infinitary version of such a construction 
(and indeed some of the basic  ideas  behind this construction)
is described in \cite{GVZ04}. However, the main goal in \cite{GVZ04} 
was to obtain upper bounds on the Betti numbers of semi-algebraic (as well as
semi-Pfaffian) sets defined by quantified formulas, and this
is achieved by bounding the Betti numbers of certain sets appearing in the
$E_1$ term of a certain (the so called ``descent'') spectral sequence which
is guaranteed to converge to the homology of the given set.

In this paper we need to be able to recover exactly (not just bound) the
Betti numbers of $\pi_\Y(S)$ from those of  $D_\Y^p(S)$.
Moreover, it is very important in our context that 
membership in the semi-algebraic set  $D_\Y^p(S)$ should be checkable
in polynomial time, given that the same is true for $S$. 
Notice that even if there exists an efficient (i.e., polynomial time) 
algorithm for checking 
membership in $S$, the same need not be true for the image $\pi_\Y(S)$. 

We will now illustrate how the construction of $D_\Y^p(S)$
connects the compact real %
polynomial hierarchy to the computation of the
Betti numbers of semi-algebraic sets, by looking at the special case of one and two quantifier alternations.

\subsubsection{Case of one quantifier}
First consider the class ${\bf \Sigma}_{\R,1}^c$. 
Consider a closed semi-algebraic
set $S \subset \Sphere^k \times \Sphere^\ell$ defined by a quantifier-free
formula $\phi(\X,\Y)$ and let 
\[
\pi_\Y: \Sphere^k \times \Sphere^\ell \rightarrow \Sphere^k
\]
be the projection map along the $\Y$ variables.

Then the formula 
\[
\Phi(\X) = \exists \Y \;\phi(\X,\Y)
\] 
is satisfied by $\x \in \Sphere^k$
if and only if 
$b_0(S_\x) \neq 0$, where $S_\x = S \cap \pi_\Y^{-1}(\x)$.
Thus, the problem of deciding the truth of $\Phi(\x)$ is reduced to computing
a  Betti number (the $0$-th) of the fiber of $S$ over $\x$. 
 
Now consider the class ${\bf \Pi}_{\R,1}^c$. Using the same notation as above
we have that the formula
\[
\Psi(\X) = \forall \Y\; \phi(\X,\Y)
\] 
is satisfied by  $\x \in \Sphere^k$
if and only if 
the formula 
\[
\neg\Psi(\X) = \exists \Y\; \neg\phi(\X,\Y)
\]
does not hold, which means, according to the previous case, that 
we have
$b_0(\Sphere^\ell \setminus S_\x) = 0$, which is equivalent to
$b_\ell(S_\x) = 1$. Notice that, as before,
the problem of deciding the truth of $\Psi(\x)$ is reduced to computing
a  Betti number (the $\ell$-th) of the fiber of $S$ over $\x$. 

\subsubsection{Case of two quantifiers}
Proceeding to a slightly more non-trivial case, 
consider the class ${\bf \Pi}_{\R,2}^c$ and let
$S \subset \Sphere^k \times \Sphere^\ell \times \Sphere^m$ be
a closed semi-algebraic set defined by a quantifier-free formula
$\phi(\X,\Y,\ZB)$ and let 
\[
\pi_{\ZB}: \Sphere^k \times \Sphere^\ell \times
\Sphere^m  \rightarrow \Sphere^k\times \Sphere^\ell
\]
be the projection map along the $Z$ variables, and
\[
\pi_\Y: \Sphere^k \times \Sphere^\ell \rightarrow \Sphere^k
\]
be the projection map along the $Y$ variables as before. Consider the
formula 
\[
\Phi(\X) = \forall \Y \; \exists \ZB \;\phi(\X,\Y,\ZB).
\]

This formula can be recast as:
\[
\Phi(\X) = \forall \Y \; (\X,\Y)\in \pi_{\ZB}(S).
\]

Thus, for any $\x\in \Sphere^k$,  $\Phi(\x)$ holds if and only if we have the following situation:
\[
\xymatrix{
 & S \ar@{^{(}->}[r] \ar[d]^{\pi_{\ZB}}& 
 \Sphere^k\times \Sphere^\ell\times \Sphere^m \ar[d]^{\pi_{\ZB}}\\
\{\x\}\times \Sphere^\ell \ar@{^{(}->}[r] \ar[dr]^{\pi_{\Y}}
& \pi_{\ZB}(S) \ar@{^{(}->}[r] \ar[d]^{\pi_{\Y}}& 
\Sphere^k\times \Sphere^\ell \ar[d]^{\pi_{\Y}}\\
& \x\in \pi_{\Y,\ZB}(S) \ar@{^{(}->}[r] & \Sphere^k
}
\]
i.e., if and only if the $\pi_\Y$ fiber $\left(\pi_{\ZB}(S)\right)_\x$ is equal to $\Sphere^\ell$.
This can be formulated in terms of Betti numbers by the condition:
$$ b_\ell\left(\left(\pi_{\ZB}(S)\right)_\x\right)=1.$$
The construction mentioned in~\eqref{eq:D} gives, for $p=\ell+1$, the existence of a semi-algebraic set $D_{{\ZB}}^{\ell+1}(S)$ such that 
$b_\ell(D_{{\ZB}}^{\ell+1}(S))=b_\ell(\pi_{\ZB}(S)).$ Fortunately, the construction of the set $D_{{\ZB}}^{\ell+1}(S)$ is compatible with taking fibers, 
so that we have, for all $\x\in \Sphere^k$,
$$ b_\ell\left(\left(\pi_{\ZB}(S)\right)_\x\right)
= b_\ell\left(D_{{\ZB}}^{\ell+1}(S)_\x\right).$$

Thus for any  $\x \in \Sphere^k$, the
truth or falsity of $\Phi(\x)$ is determined by a certain Betti number of the
fiber $D_{{\ZB}}^{\ell+1}(S)_\x$ over $\x$ of a certain
semi-algebraic set  $D_{{\ZB}}^{\ell+1}(S)$ 
which can be constructed efficiently
in terms of the set $S$. The idea behind the proof of the main theorem
is a recursive application of the above argument in case when the number
of quantifier alternations is larger (but still bounded by some constant)
while keeping track of the growth in the sizes of the intermediate formulas
and also the number of quantified variables.

The rest of the paper is organized as follows. In Section \ref{sec:ingredients}
we fix notation, and prove the topological results needed for the proof of
the two main theorems. 
In Section \ref{sec:saconstructions} we describe the semi-algebraic 
construction of the sets $D_\Y^p(S)$ alluded to above and prove 
its important  properties. 
We prove the main theorems in Section \ref{sec:proof}.

\section{Topological Ingredients}
\label{sec:ingredients}
We first fix a notation.

\begin{notation}
\label{not:simplex}
For each $p \geq 0$ we denote
\[
\Delta^p = \{(t_0,\ldots,t_p)\mid t_i \geq  0, 0 \leq i \leq p, \sum_{i=0}^p 
t_i = 1
\}
\]
the standard $p$-simplex.
\end{notation}

We now describe some constructions in algebraic topology which 
will be useful later in the paper.

\subsection{Properties of the join}
We first recall the definition of  the join of two topological spaces $X$ and
$Y$.

\begin{definition}
\label{def:twofoldjoin}
The join $J(X,Y)$  of two topological spaces $X$ and $Y$ is defined by
\begin{equation}
\label{eqn:definitionoftwofoldjoin}
J(X,Y) \defeq X\times Y 
\times \Delta^1/\sim, 
\end{equation}
where
\[
(x,y,t_0,t_1) \sim (x',y',t_0,t_1)
\] 
if
$t_0 = 1,x = x'$ or  $t_1=1, y= y'$.
\end{definition}

 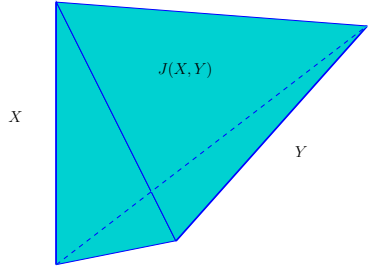
\begin{figure}[hbt]
         \centerline{
           \scalebox{0.5}{
             \input{join.pstex_t}
             }
           }
         \caption{Join of two segments}
         \label{fig:join}
 \end{figure}

Intuitively, $J(X,Y)$ is obtained by joining each point of $X$ with
each point of $Y$ by a unit interval (see Figure \ref{fig:join}).

\begin{example}
\label{eg:joinoftwospheres}
It is easy to check from the above definition that
the join, $J(\Sphere^m,\Sphere^n)$, of two spheres
is again (homeomorphic to) a sphere, namely $\Sphere^{m+n+1}$.
\end{example}

By iterating the above definition with the same space $X$ we obtain 
\begin{definition}
\label{def:pfoldjoin}
For $p \geq 0$ the $(p+1)$-fold join $J^p(X)$  of $X$ is  
\begin{equation}
\label{eqn:definitionofjoin}
J^p(X) \defeq \underbrace{X\times\cdots\times X}_{(p+1)\mbox{ times }}
\times \Delta^p/\sim, 
\end{equation}
where
\[
(x_0,\ldots,x_p,t_0,\ldots,t_p) \sim (x_0',\ldots,x_p',t_0,\ldots,t_p)
\] 
if for each $i$ with $t_i \neq 0$, $x_i = x_i'$.
\end{definition}

\begin{example}
\label{eg:joinofzerodimspheres}
Using Example \ref{eg:joinoftwospheres} it is easy to see  
that the $(p+1)$-fold
join, $J^p(\Sphere^0)$, of the zero dimensional sphere
is homeomorphic to $\Sphere^p$.
\end{example}

We will need the fact that the iterated join of a topological space is 
highly connected. In order to make this statement precise we first define

\begin{definition}[$p$-equivalence]
\label{def:p-equivalence}
A map $f: A \rightarrow B$ between two topological spaces is called a 
\bfdef{$p$-equivalence} if the induced homomorphism 
\[
f_*: \HH_i(A) \rightarrow \HH_i(B)
\]
is an isomorphism for all $0 \leq i < p$, and an epimorphism for $i=p$,
and we say that $A$ is \bfdef{$p$-equivalent} to $B$. (Note that $p$-equivalence
is not an equivalence relation
: e.g. for any $p\geq 0$, the map taking $\Sphere^p$ to a point is a $p$-equivalence, 
but no map from a point into $\Sphere^p$ is one.)
\end{definition}

Observe from Example \ref{eg:joinofzerodimspheres} that 
$J^p(\Sphere^0) \cong \Sphere^p$ is $p$-equivalent to a point.
In fact, this holds much more generally and we have that
 
\begin{theorem}
\label{the:pjoinconnectivity}
Let $X$ be 
a semi-algebraic set.
Then, the $(p+1)$-fold join
$J^p(X)$ is $p$-equivalent to a point.
\end{theorem}

\begin{proof}
This is classical (see for instance 
\cite[Proposition 4.4.3]{Matousek_book2}).
\end{proof}

\subsection{Join over a map}
In our application we need the join construction over certain class of maps
(to be specified later). We first recall the notion of a fibered product
of a topological space.

\begin{notationdefinition}
\label{notdef:fiberproduct}
Let $f:A \rightarrow B$ be a map
between topological spaces $A$ and $B$. 
For each $p \geq 0$, 
We denote by $W_f^p(A)$ the \bfdef{$(p+1)$-fold fiber product} of $A$ over $f$. 
In other words
\[
W_f^p(A) = \{(x_0,\ldots,x_p) \in A^{p+1}
\mid
f(x_0) = \cdots = f(x_p)
\}.
\]
\end{notationdefinition}

\begin{definition}[Topological join over a map]
\label{def:joinoveramap1}
Let $f:A \rightarrow B$ be a map
between topological spaces $A$ and $B$. 
For $p \geq 0$ the \bfdef{$(p+1)$-fold join} $J^p_f(A)$  of $A$ over $f$ 
is  
\begin{equation}
\label{eqn:definitionofjoin1}
J^p_f(A) \defeq W_f^p(A) \times \Delta^p/\sim, 
\end{equation}
where
\[
(x_0,\ldots,x_p,t_0,\ldots,t_p) \sim (x_0',\ldots,x_p',t_0,\ldots,t_p)
\] 
if for each $i$ with $t_i \neq 0$, $x_i = x_i'$.
\end{definition}

We now impose certain conditions on the map $f$.
\subsection{Compact Coverings}
Recall that we call $K \subset \R^n$ a semi-algebraic compact 
if it is a closed and bounded semi-algebraic set. 
\begin{notation}
For any semi-algebraic $A \subset \R^n$, we denote by 
$\K(A)$ the
collection of all semi-algebraic compact subsets of $A$.
\end{notation}

\begin{definition}
\label{def:compactcovering}
Let $f: A \to B$ be a semi-algebraic map. We say that $f$ \bfdef{covers
semi-algebraic compacts} if for any $L\in \K(f(A))$, there
exists $K\in \K(A)$ such that $f(K)=L.$
\end{definition}

The following theorem relates the topology of $J^p(A)$ to that of
the image of $f$ in the case when $f$ covers semi-algebraic
compacts and is crucial for what follows.

\begin{theorem}
\label{the:compactcovering}
Let $f: A \to B$ be a semi-algebraic map that covers
semi-algebraic compacts. 
Then for every $p \geq 0$, the map $f$ induces a $p$-equivalence 
$J(f): J_f^p(A) \to f(A)$.
\end{theorem}

\begin{proof}

We begin with the case $A\in \K(\R^n)$. 
Let $J(f): J_f^p(A) \to f(A)$ be the map given by 
\[J(f)(x_0, \dots,x_p,t_0,\dots, t_p)=f(x_0). \]
The map $J(f)$ is well defined since $(x_0,\dots, x_p)\in W^p_f(A)$,
and is closed since $J_f^p(A)$ is a semi-algebraic compact.
Moreover, the fibers of $J(f)$ are $p$-equivalent to a point
by Theorem \ref{the:pjoinconnectivity}.
 
Thus, by the Vietoris-Begle theorem 
\cite[Theorem 2]{BWW06}
the map $J(f)$ induces
isomorphisms 
$$ J(f)_*: \HH_i(J_f^p(A)) \to \HH_i(f(A));$$
for $0 \leq i < p.$ Note that in the case $\R\neq \mathbb{R}$, the
validity of the Vietoris-Begle theorem can be seen as a corollary of
the existence of a semi-algebraic co-homology  that satisfies the
Eilenberg-Steenrod axioms for a \v{C}ech theory (see \cite{EP}). 

\bigskip

In the general case, consider $K_1 \subset K_2$ two semi-algebraic
compacts in $\K(A)$. The inclusion gives rise to the following diagram,
\[
\xymatrix{
J_f^p(K_1) \ar@{^{(}->}[r]^{i} \ar[d]^{J(f|_{K_1})}& 
J_f^p(K_2)\ar[d]^{J(f|_{K_2})} \\
f(K_1)\ar@{^{(}->}[r]^{j} & f(K_2) \\
}
\]
where the vertical maps are $p$-equivalence by the previous
case. We have a similar diagram at the homology level; if we take the
direct limit as $K$ ranges in $\K(A)$, we obtain the following:

\[
\xymatrix{
\dlim \HH_i(J_f^p(K)) \ar[r]^{\cong} \ar[d]^{\dlim J(f|_K)}& 
\HH_i(J_f^p(A))\ar[d]^{J(f)} \\
\dlim \HH_i(f(K))\ar[r]^{\cong} & \HH_i(f(A)) \\
}
\]
The isomorphism on the top level comes from the fact that homology and
direct limit commute~\cite[Theorem~7 on p.~162]{Spanier}, along with the fact that for a
semi-algebraic set, one can compute the homology using chains
supported exclusively on semi-algebraic
compacts~\cite{Delfs-Knebusch}. For the bottom isomorphism, we need
the additional fact that since 
$f$ covers
semi-algebraic compacts, 
the map $J(f): J_f^p(A) \to f(A)$ also covers semi-algebraic compacts.
Hence, we have $$ \dlim \{\HH_i(f(K)) \mid K\in\K(A)\}
=\dlim \{\HH_i(L) \mid L\in \K(B)\}. $$ Since each $J(f|_K)$ 
is a $p$-equivalence, 
the vertical homomorphisms are isomorphisms for $0 \leq i < p$, and
an epimorphisms for $i=p$.
\end{proof}

\begin{remark}
\label{rem:restrictiontocompacts}
Theorem~\ref{the:compactcovering} requires that the map $f$ 
covers semi-algebraic compacts.
This condition is
satisfied for a projection in the case the set $A$ 
is either open or compact.
Note also that Theorem~\ref{the:compactcovering} 
is not true without the
assumption that $f$ covers semi-algebraic compacts, 
which is why, in this paper, we restrict our attention to 
the \emph{compact} polynomial hierarchy.
\end{remark}

\subsection{Alexander-Lefshetz duality}
We will also need the classical Alexander-Lefshetz duality theorem
in order to relate the Betti numbers of a compact semi-algebraic subset
of a sphere to those of its complement.

\begin{theorem}[Alexander-Lefshetz duality]
\label{the:alexanderduality}
Let $K \subset \Sphere^n$ be a compact semi-algebraic subset with $n \geq 2$.
Then 
\begin{eqnarray*}
b_0(K) &=& 1 + b_{n-1}(\Sphere^n - K) - b_n(\Sphere^n - K),\\
b_i(K) &=& b_{n-i-1}(\Sphere^n - K), \;\;1 \leq i \leq n-2, \\
b_{n-1}(K) &=& b_0(\Sphere^n - K) - 1 + 
\max(1 - b_0(\Sphere^n - K),0),\\
b_{n}(K) &=& 1 - \min(1,b_0(\Sphere^n - K)).
\end{eqnarray*}
\end{theorem}
\begin{proof}
The Lefshetz duality theorem \cite[Theorem~19 on p.~297]{Spanier} gives for each $i, 0 \leq i \leq n$,
\[
b_i(\Sphere^n - K) = b_{n-i}(\Sphere^n, K).
\]
The theorem now follows from the long exact sequence of homology,
\[
\cdots \rightarrow \HH_i(K) \rightarrow \HH_i(\Sphere^n) \rightarrow 
\HH_i(\Sphere^n,K) \rightarrow \HH_{i-1}(K) \rightarrow \cdots
\] 
after noting that $\HH_i(\Sphere^n) = 0, i \neq 0,n$ and 
$\HH_0(\Sphere^n) = \HH_n(\Sphere^n) = \Z$.
\end{proof}

\section{Semi-algebraic Constructions}
\label{sec:saconstructions}
In this section we describe the semi-algebraic construction that lies
at the heart of the proof of our main theorem, and prove its important
properties. 

Let $S \subset \Sphere^k \times \Sphere^\ell$ 
be a subset 
defined by a first-order formula $\Phi(\X,\Y)$, 
and let  $\pi_\Y$ denote the projection along the $\Y$ co-ordinates.
Note that $\X$ and $\Y$ are the \emph{free} variables in $\Phi$, the formula being considered may have any number of quantified blocks of variables $(\ZB^1, \dots, \ZB^\omega).$ This will be addressed explicitly in Lemma~\ref{lem:main}.

We now define semi-algebraic sets having the same homotopy type
as the join space $J_{\pi_\Y}(S)$  in the case when $S$ is a closed 
(respectively open) semi-algebraic  subset of  $\Sphere^k \times \Sphere^\ell$.

\begin{notationdefinition}
Let $S \subset \Sphere^k \times \Sphere^\ell$,
$\Phi(\X,\Y)$, 
and $\pi_\Y$ be as above. If $S$ is \emph{closed}, 
we denote by $D_{\Y,c}^p(S)$  the semi-algebraic set 
defined by
\begin{eqnarray}
\label{eqn:definitionofjoin2}
D_{\Y,c}^p(S) &\defeq& \{ (u,\x,\y^0,\ldots,\y^p,\tb) \mid 
\x \in \Sphere^k, \tb \in \Delta^p, \nonumber
\\
&&
\text{ for each } i, \;0 \leq i \leq p, \y^i \in \Ball^{\ell+1},
(t_i = 0)  \vee \Phi(\x,\y^i),   \\
&&
u^2 + |\x|^2 + \sum_{i=0}^p |\y^i|^2
+ |\tb|^2 = p+4, \mbox{and } u \geq 0
\}.
\nonumber
\end{eqnarray}

Notice that $D_{\Y,c}^p(S)$
is a closed semi-algebraic subset of the upper hemisphere of the
sphere $\Sphere^N(0,p+4)$,  where $N = (k+1) + (p+1)(\ell+2).$

We will denote by $D_{\Y,c}^p(\Phi)$ the first-order formula 
defining the 
semi-algebraic set $D_{\Y,c}^p(S)$, namely

\begin{equation}
\label{def:defofDclosed}
D_{\Y,c}^p(\Phi)
\defeq
\Theta_1(T) \wedge  \Theta_2(\X,\Y^0,\ldots,\Y^p,\TB)
\wedge \Theta_3(U_0,\X,\Y^0,\ldots,\Y^p,\TB)
\end{equation}
where
\begin{eqnarray*}
\Theta_1
&\defeq& (\bigwedge_{i=0}^{p} T_i \geq 0) 
\wedge (\sum_{i=0}^{p} T_i = 1), \\
\Theta_2
&\defeq&
((|\X|^2 = 1) \bigwedge_{i = 0}^{p} ((|\Y^i|^2 \leq 1) \wedge ((T_i = 0) \vee \Phi(\X,\Y^i))),\\
\Theta_3
&\defeq& (U_0^2 + |\X|^2 + \sum_{i=0}^p |\Y^i|^2
+ |\TB|^2 = p+4)\wedge (U_0 \geq 0). \\
\end{eqnarray*}
\end{notationdefinition}

We have a similar construction in case $S$ is an open subset of
$\Sphere^k \times \Sphere^\ell$. In this case we thicken the various
faces of the standard simplex $\Delta^p$ 
(see Figure \ref{fig:thickened})
so that they become convex open 
subsets of $\R^{p+1}$, but maintaining the property that a subset of these 
thickened faces have a non-empty intersection if and only if the closures of
the corresponding faces in $\Delta^p$ had a non-empty intersection. 
In this way we ensure that our construction produces an open subset of a 
sphere, while having again the homotopy type of the join space. 

 \begin{figure}[hbt]
         \centerline{
           \scalebox{0.5}{
             \input{thickened.pstex_t}
             }
           }
         \caption{Thickening of the simplex $\Delta^1$.}
         \label{fig:thickened}
 \end{figure}
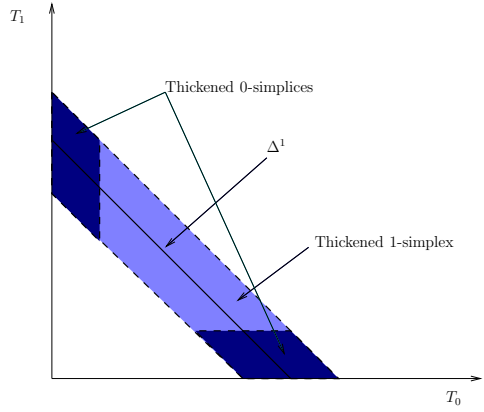

\begin{notationdefinition}
Let $S \subset \Sphere^k \times \Sphere^\ell$ 
be an open subset 
defined by the first-order formula $\Phi(\X,\Y)$, 
and let  $\pi_\Y$ denote the projection along the $\Y$ co-ordinates.

We will denote by $D_{\Y,o}^p(\Phi)$ the following first-order formula.

\begin{equation}
\label{def:defofDopen}
D_{\Y,o}^p(\Phi)
\defeq
\Theta_1(\TB) \wedge \Theta_2(\X,\Y^0,\ldots,\Y^p,\TB)
\wedge \Theta_3(U_0,\X,\Y^0,\ldots,\Y^p,\TB)
\end{equation}
where
\begin{eqnarray*}
\Theta_1
&\defeq& (\bigwedge_{i=0}^{p} T_i >  0) 
\wedge (1 - \frac{1}{2(p+1)} < \sum_{i=0}^{p} T_i  < 1 + \frac{1}{2(p+1)}), \\
\Theta_2
&\defeq&
 \bigwedge_{i = 0}^{p} ((|\Y^i|^2 < 3/2)\wedge ((T_i < \frac{1}{2(p+1)} 
\vee \Phi_{+}(\X,\Y^i))),\\
\Theta_3
&\defeq& (U_0^2 + |\X|^2 + \sum_{i=0}^p |\Y^i|^2
+ |\TB|^2 = 2p+4) \wedge (U_0 > 0) \\
\end{eqnarray*}
and
\[
\Phi_{+}(\X,\Y) \defeq (1/2 < |\X|^2 <  3/2) \wedge (1/2 < |\Y|^2 <  3/2)
\wedge \Phi(\X/|\X|,\Y/|\Y|).
\] 
We will denote by $D_{\Y,o}^p(S)$
the semi-algebraic set defined by  $D_{\Y,o}^p(\Phi)$.
Notice that $D_{\Y,o}^p(S)$ is an open subset of the upper hemisphere
of the sphere $\Sphere^N(0,2p+4)$, where
$N = (k+1)+(p+1)(\ell+2)$.
\end{notationdefinition}

We now prove some important properties of the sets 
$D_{\Y,c}^p(S),D_{\Y,o}^p(S)$ defined above as well as of the formulas
$D_{\Y,c}^p(\Phi),D_{\Y,o}^p(\Phi)$ defining them.

\begin{proposition}[Polynomial time computability]
\label{prop:polytime}
Let 
\[
(\Phi_n(X_0,\ldots,X_{k(n)},Y_0,\ldots,Y_{\ell(n)})_{n>0}, 
k,\ell = n^{O(1)})
\]
be a sequence of quantifier-free first order formulas such that
for each $n> 0$, 
$\Phi_n$ defines a closed (respectively open) semi-algebraic subset $S_n$  
of $\Sphere^{k(n)} \times \Sphere^{\ell(n)}$
and suppose that $(S_n)_{n > 0} \in \mathbf{P}_\R$.
 
Then the sequence 
$(D_{\Y,c}^p(\Phi_n))_{n >0}$
(respectively, $(D_{\Y,o}^p(\Phi_n))_{n >0}$) belongs to 
$\mathbf{P}_\R$ as well.
\end{proposition}

\begin{proof}
Clear from the construction of the formulas $(D_{\Y,c}^p(\Phi_n))_{n >0}$
(respectively $(D_{\Y,o}^p(\Phi_n))_{n >0}$).
\end{proof}

We now prove an important topological property of the semi-algebraic
sets $D_{\Y,c}^p(S),D_{\Y,o}^p(S)$ defined above.

\begin{proposition}[Homotopy equivalence to the join]
Let $S \subset \Sphere^k \times \Sphere^\ell$ 
be a closed  (respectively, open)  subset of $\Sphere^k \times \Sphere^\ell$ 
defined by a first-order formula $\Phi(\X,\Y)$, 
and let  $\pi_\Y$ denote the projection along the $\Y$ co-ordinates.
Then for all $p\geq 0$,
$J^p_{\pi_\Y} (S)$ is homotopy equivalent to  $D_{\Y,c}^p(S)$ 
(respectively, $D_{\Y,o}^p(S)$ ).
\end{proposition}

\begin{proof}
Suppose 
$S$ is a closed subset of $\Sphere^k \times \Sphere^\ell$ and let
\[
g: D_{\Y,c}^p(S) \rightarrow J^p_{\pi_\Y}(S)
\]
be the map which takes a point $(u,\x,\y^0,\ldots,\y^p,\tb) 
\in D_{\Y,c}^p(S)$
to the equivalence class represented by the point 
 $((\x,\y^0),\ldots,(\x,\y^p),\tb)$ in $ J_{\pi_\Y}^p(S)$.
From the definition of the spaces 
$D_{\Y,c}^p(S)$ and $J_{\pi_\Y}^p(S)$, we have that
the inverse image under $g$ of a point represented by 
 $((\x,\y^0),\ldots,(\x,\y^p),\tb)$ in $ J_{\pi_\Y}^p(S)$ is given by
$$
\displaylines{
g^{-1}(((\x,\y^0),\ldots,(\x,\y^p),\tb)) = 
\{ (u,\x,\z^0,\ldots,\z^p,\tb) \mid \mbox{ for each } i, 0 \leq i \leq p, \cr
\z^i \in \Ball^{\ell+1} \mbox{ and } \z^i = \y^i \mbox{ if } t_i \neq 0,
u^2 + |\x|^2 + \sum_{i=0}^p |\z^i|^2
+ |\tb^2| = p+4, u \geq 0
\}.
}
$$
It is easy to see from the above formula that 
the inverse image under $g$ of each point of $J^p_{\pi_\Y}(S)$ is homeomorphic
to a product of balls and hence contractible. 
The proposition now follows from the
Vietoris-Begle theorem. 

The open case is  proved analogously 
after an infinitesimal retraction 
reducing it to the closed case.
\end{proof}

As an immediate corollary we obtain
\begin{corollary}
\label{cor:joinhomotopy}
Let $S \subset \Sphere^k \times \Sphere^\ell$ 
be a closed  (respectively, open)  subset of $\Sphere^k \times \Sphere^\ell$ 
defined by a first-order formula $\Phi(\X,\Y)$, 
and let  $\pi_\Y$ denote the projection along the $\Y$ co-ordinates.
Then for all $p\geq 0$,
$D_{\Y,c}^p(S)$ (respectively, $D_{\Y,o}^p(S)$)  is 
$p$-equivalent to $\pi_\Y(S)$,
and
\[
b_i(D_{\Y,c}^p(S)) = b_i(\pi_\Y(S))
\]
(respectively,
$\displaystyle{b_i(D_{\Y,o}^p(S)) = b_i(\pi_\Y(S))}$
)
for $0 \leq i < p$. 
\end{corollary}

\begin{proof}
Since $S$ is either an open or closed subset of 
$\Sphere^k \times \Sphere^\ell$ it is clear that
the projection map
$\pi_\Y$ covers semi-algebraic compacts. Now apply Theorem 
\ref{the:compactcovering}. 
\end{proof}

We now show how the formulas $D_{\Y,c}^p(\Phi)$ and 
$D_{\Y,o}^p(\Phi)$ can be rewritten when the formula $\Phi$ 
involves quantified blocks of variables.

\begin{lemma}
\label{lem:main}
Suppose the first-order formula $\Phi(\X,\Y)$ is of the form
\[
\Phi\defeq (Q_1 \ZB^1 \in \Sphere^{k_1})( Q_2 \ZB^2 \in \Sphere^{k_2}) \ldots
(Q_\omega \ZB^\omega \in \Sphere^{k_\omega}) \Psi(\X,\Y,\ZB^1,\ldots,\ZB^\omega)
\]
with $Q_i \in \{\exists,\forall\}$, and $\Psi$ a quantifier-free
first order formula.
Let $\pi_\Y$ denote the projection  along the $Y$ coordinates. 
Then, for each $p \geq 0$ the formula 
\[
D_{\Y,*}^p(\Phi)(\X,\Y^0,\ldots,\Y^p,\TB)
\] 
(where $*$ denotes either $c$ or $o$)
is equivalent to the formula
$$
\displaylines{
\bar{D}_{\Y,*}^p(\Phi) \defeq
(Q_1 \ZB^{1,0} \in \Sphere^{k_1},\ldots, Q_1 \ZB^{1,p} \in \Sphere^{k_1}) \cr
(Q_2 \ZB^{2,0} \in \Sphere^{k_2},\ldots, Q_2 \ZB^{2,p} \in \Sphere^{k_2}) \cr
\vdots \cr
(Q_\omega \ZB^{\omega,0} \in \Sphere^{k_\omega},\ldots, Q_\omega \ZB^{\omega,p} \in \Sphere^{k_\omega}) \cr
(D_{\Y,*}^p(\Psi)(\X,\Y^0,\ldots,\Y^p,\ZB^{1,0},\ldots,\ZB^{\omega,p},T_0,\ldots,T_p)),
}
$$
where $\Y^i = (Y^i_0,\ldots,Y^i_\ell)$ and
$\ZB^{j,i} = (Z^{j,i}_0,\ldots,Z^{j,i}_{k_j})$
for $0 \leq i \leq p, 1 \leq j \leq \omega$, and
$\pi_{\Y}$ is the projection along the $\Y$ co-ordinates.
\end{lemma}

\begin{proof}
It follows from the structure of the formula 
$D_{\Y,*}^p(\Phi)(\X,\Y^0,\ldots,\Y^p,\TB)$ that the
inner most quantifiers can be pulled outside at the cost of introducing
$(p+1)$ copies of the quantified  variables.
\end{proof}

\section{Proof of the main theorems}
\label{sec:proof}

\begin{notation}
Let $\Phi(\X)$ be a first-order formula with free variables 
$\X=(X_0, \dots, X_k)$. We let $\RR(\Phi(\X))$ denote the \emph{realization} of the formula $\Phi$,
$$ \RR(\Phi(\X))=\{\x \in \R^{k+1} \mid \Phi(\x)\}.$$
\end{notation}

We are now in a position to prove Theorem \ref{the:main}.
The proof depends on the following key proposition. 
(Note that we are going use Proposition \ref{prop:main}  
in the special case when the set of variables $Y$ is empty).

\begin{proposition}
\label{prop:main}
Let $m(n), k(n), k_1(n),\ldots,k_\omega(n)$ be polynomials,  
and let 
\[
(\Phi_n(\X,\Y))_{n>0}
\] 
be a sequence of formulas 
$$ \Phi_n(\X,\Y) \defeq  
(Q_1 \ZB^1 \in \Sphere^{k_1})  \cdots 
(Q_\omega \ZB^\omega\in \Sphere^{k_\omega})
\phi_n(\X,\Y,\ZB^{1},\ldots,\ZB^{\omega}),$$
having free variables $(\X,\Y) = (X_0,\ldots,X_{k(n)},Y_0,\ldots,Y_{m(n)})$,
with 
\[
Q_1,\ldots,Q_\omega \in \{\exists,\forall\}, Q_i \neq Q_{i+1},
\]
and $\phi_n$ a  quantifier-free formula
defining a closed (respectively open) semi-algebraic subset of 
$$
\Sphere^k \times \Sphere^m \times \Sphere^{k_1} \times \cdots \times
\Sphere^{k_\omega}.
$$
Suppose that the sequence 
$\left(\RR(\phi_n(\X,\Y,\ZB^{1},\ldots,\ZB^{\omega})\right)_{n > 0} \in
\mathbf{P}_\R$.

Then,  
\begin{enumerate}
\item
there exists $N = N(m)$ a fixed polynomial in $m$,
variables $\V=(V_0,\ldots,V_{N})$
and a sequence of
quantifier-free first order formulas 
\[
(\Theta_n(\X,\V))_{n>0}
\]
such that for each $\x \in \Sphere^{k(n)}$,
$\Theta_n(\x,\V)$ describes a closed (respectively open)
semi-algebraic subset $T_n$ of  $\Sphere^N$
and the sequence
$(T_n)_{n > 0} \in \mathbf{P}_\R$;
\item
there exists a sequence of polynomial-time computable maps
\[
\left(F_n: \Z[T]_{\leq N}\rightarrow \Z[T]_{\leq m}\right)_{n > 0}
\]
(where for any $p \geq 0$, 
$\Z[T]_{\leq p}$ denotes polynomials of degree at most $p$) 
such that the Poincar\'e polynomials of the fibers over $\x$ verify
\[
P_{\RR(\Phi_n(\x,\Y))} =  F_n(P_{\RR(\Theta_n(\x,\V))}).
\]
\end{enumerate}
\end{proposition}

\begin{proof}
The proof is by an induction on $\omega$.
We assume that the formula $\phi_n$ defines a closed semi-algebraic set. 
The open case can be handled analogously.
 
If $\omega=0$ then we let for each $n >0$, 
$N = m(n)$, $\Theta_n = \Phi_n(\X,\V)$, 
and 
$F_n$ to be the identity map.
Since $\Phi_n$ is quantifier-free in this case, so is $\Theta_n$,
and we have trivially 
\[
P_{\RR(\Phi_n(\x,\Y))} =  F_n(P_{\RR(\Theta_n(\x,\V))}).
\]

If $\omega > 0$, we have the following two cases.
\begin{enumerate}
\item
Case 1, $Q_1 = \exists$: 
In this case consider the
sequence of formulas 
$\Phi'_n \defeq \bar{D}_{\pi_{\ZB^1},c}^m(\Psi_n)$ (cf. Lemma~\ref{lem:main}), 
where $\Psi_n$ is the following formula with free variables $\Y,\ZB^1$
\begin{equation}\label{eq:Psin}
  \Psi_n(\X,\Y,\ZB^1) \defeq  (Q_2 \ZB^2 \in \Sphere^{k_2})  \cdots 
(Q_\omega \ZB^\omega\in \Sphere^{k_\omega})
\phi_n(\X,\Y,\ZB^{1},\ldots,\ZB^{\omega}).
\end{equation}

The formula $\bar{D}_{{\ZB^1},c}^m(\Psi_n)$ is given by
$$
\displaylines{
\bar{D}_{\pi_{\ZB^1},c}^m(\Psi_n) \defeq
(Q_2 \ZB^{2,0} \in \Sphere^{k_2},\ldots, Q_2 \ZB^{2,m} \in \Sphere^{k_2}) \cr
(Q_3 \ZB^{3,0} \in \Sphere^{k_3},\ldots, Q_3 \ZB^{3,m} \in \Sphere^{k_3}) \cr
\vdots \cr
(Q_\omega \ZB^{\omega,0} \in \Sphere^{k_\omega},\ldots, Q_\omega \ZB^{\omega,m} \in \Sphere^{k_\omega}) \cr
(D_{{\ZB^1},c}^m(\phi_n)(\X,\Y,\ZB^{1,0},\ldots,\ZB^{1,m},
\ZB^{2,0},\ldots,\ZB^{\omega,m},T_0,\ldots,T_m)).
}
$$
Note that the quantifier-free inner formula in the above expression, 
\[
D_{{\ZB^1},c}^m(\phi_n)(\X,\Y,\ZB^{1,0},\ldots,\ZB^{1,m},\ZB^{2,0},\ldots,\ZB^{\omega,m},T_0,\ldots,T_m)
\]
has 
\[
N = \sum_{j=1}^{\omega} (k_j+1)(m+1) + 2(m+1) = n^{O(1)}
\]
free variables, and it is clear 
that the sequence $(\RR(D_{{\ZB^1},c}^m(\phi_n)))_{n > 0} \in \mathbf{P}_\R$
(by Proposition \ref{prop:polytime}).

Moreover, 
the formula $\bar{D}_{\pi_{\ZB^1},c}^m(\Psi_n)$
has one less quantifier alternation than the formula $\Phi_n$.

We now apply the proposition inductively to obtain
a sequence
$(\Theta_n)_{n > 0}$ with $(\RR(\Theta_n))_{n > 0} \in \mathbf{P}_\R$, 
and polynomial-time computable maps
$(G_n)_{n > 0}$.
By inductive hypothesis we can suppose that
for each $i, 0 \leq i \leq m$ we have for each $\x \in \Sphere^{k(n)}$ 
\[ 
P_{\RR(\bar{D}_{{\ZB^1},c}^m(\Psi_n(\x,\Y,\ZB^1)))} = 
G_{n}(P_{\RR(\Theta_n(\x,\cdot))}).
\]

But by Corollary \ref{cor:joinhomotopy}, we have that for $0 \leq i \leq m$,
\[
b_i(\RR(\Phi_n(\x,\Y))) = b_i(\pi_{\ZB^1}(\RR(\Psi_n(\x,\Y,\ZB^1)))) = 
b_i(\RR(\bar{D}_{{\ZB^1},c}^m(\Psi_n(\x,\Y,\ZB^1)))). 
\]
We set 

\[
F_{n}= \mathrm{Trunc}_m \circ G_{n},
\]
where $\mathrm{Trunc}_m:\Z[T] \rightarrow \Z[T]_{\leq m}$ 
denotes the operator that truncates polynomials to ones of
degree at most $m$; in other words $\mathrm{Trunc}_m$ is defined by
\[
\mathrm{Trunc}_m(\sum_{0 \leq i} a_iT^i) \defeq  
\sum_{0 \leq i \leq m} a_i T^i.
\]
This
completes the induction in this case.

\item
Case 2, $Q_1 = \forall$: 
In this case consider the
sequence of formulas 
$\Phi'_n \defeq \bar{D}_{{\ZB^1},o}^m(\neg\Psi_n)$ (cf. Lemma \ref{lem:main}),
where $\Psi_n$ is 
defined as in the previous case (Eqn~\eqref{eq:Psin}).

We now apply the proposition inductively as above to obtain
a sequence
$(\Theta_n)_{n > 0}$ with $(\RR(\Theta_n))_{n > 0} \in \mathbf{P}_\R$, 
and polynomial-time computable maps
$(G_n)_{n >0}$.
By inductive hypothesis we can suppose that
for each $\x \in \Sphere^n$ 
we have 

\[ 
P_{\RR(\bar{D}_{{\ZB^1},o}^m(\neg \Psi_n(\x,\Y,\ZB^1)))} = 
G_{n}(P_{\RR(\Theta_n(\x,\cdot))}).
\]

By Corollary \ref{cor:joinhomotopy}, we have  for $0 \leq i \leq m$,
\[
b_i(\Sphere^m \setminus \RR(\Phi_n(\x,\Y)) = 
b_i(\pi_{\ZB^1}(\RR(\neg \Psi_n(\x,\Y,\ZB^1))) = 
b_i(\RR(\bar{D}_{{\ZB^1},o}^m(\neg \Psi_n(\x,\Y,\ZB^1))).
\]
The set $K = \RR(\Phi_n(\x,\Y))$ is a semi-algebraic compact, so by
Alexander-Lefshetz duality (Theorem~\ref{the:alexanderduality}), we have
\begin{eqnarray*}
b_0(K) &=& 1 + b_{m-1}(\Sphere^m - K) - b_m(\Sphere^m - K),\\
b_i(K) &=& b_{m-i-1}(\Sphere^m - K), \;\;1 \leq i \leq m-2, \\
b_{m-1}(K) &=& b_0(\Sphere^m - K) - 1 + 
\max(1 - b_0(\Sphere^n - K),0),\\
b_{m}(K) &=& 1 - \min(1,b_0(\Sphere^m - K)).
\end{eqnarray*}

We set 
$F_n$ defined by
\begin{eqnarray*} 
F_{n,0}(P) &=& 1 + G_{n,m-1}(P)- G_{n,m}(P),\\
F_{n,i}(P) &=& G_{n,m-i-1}(P), \;\;1 \leq i \leq m-2, \\
F_{n,m-1}(P) &=& G_{n,0}(P) - 1 + 
\max(1 - G_{n,0}(P),0),\\
F_{n,m}(P) &=& 1 - \min(1,G_{n,0}(P))
\end{eqnarray*}
where we denote by $F_{n,i}(P)$ (respectively $G_{n,i}(P)$) the coefficient
of $T^i$ in  $F_n(P)$  (respectively $G_{n}(P)$).
This completes the induction in this case as well.
\end{enumerate}
\end{proof}

\begin{proof}[Proof of Theorem \ref{the:main}] 
Follows immediately from Proposition \ref{prop:main} in the special case
when the set of variables $\Y$ is empty. In this case the 
sequence of formulas $(\Phi_n)_{n>0}$ 
correspond to a language in the polynomial
hierarchy and for each $n$, $\x = (x_0,\ldots,x_{k(n)}) \in S_n 
\subset\Sphere^{k(n)}$ 
if and only if
\[
F_{n}(P_{\RR(\Theta_n(\x,\cdot))})(0) >  0
\]
and 
this last condition can be checked in polynomial time with advice from the
class $\#\mathbf{P}^\dagger_R$.
\end{proof}

\begin{remark}
It is interesting to observe that
in complete analogy with the proof of the classical Toda's 
theorem  the proof of Theorem \ref{the:main} also 
requires just one call to the oracle at the end.
\end{remark}

\begin{proof}[Proof of Theorem \ref{the:main2}]
Follows from the proof of 
Proposition \ref{prop:main} since the 
the formula
$\Theta_n$ is clearly computable in polynomial time from the given formula
$\Phi_n$ as long as the number of quantifier alternations $\omega$ is 
bounded by a constant.
\end{proof}

\bibliographystyle{amsplain}
\bibliography{master}
\end{document}

%% file: sharpP.pstex_t
\begin{picture}(0,0)%
\includegraphics{sharpP.pstex}%
\end{picture}%
\setlength{\unitlength}{3947sp}%
\begingroup\makeatletter\ifx\SetFigFont\undefined%
\gdef\SetFigFont#1#2#3#4#5{%
  \reset@font\fontsize{#1}{#2pt}%
  \fontfamily{#3}\fontseries{#4}\fontshape{#5}%
  \selectfont}%
\fi\endgroup%
\begin{picture}(7350,5480)(1501,-5219)
\put(1501,-661){\makebox(0,0)[lb]{\smash{{\SetFigFont{12}{14.4}{\familydefault}{\mddefault}{\updefault}{\color[rgb]{0,0,0}${\huge  k^m}$}%
}}}}
\put(8851,-286){\makebox(0,0)[lb]{\smash{{\SetFigFont{12}{14.4}{\familydefault}{\mddefault}{\updefault}{\color[rgb]{0,0,0}${\huge L_{m+n}}$}%
}}}}
\put(7501,-5161){\makebox(0,0)[lb]{\smash{{\SetFigFont{12}{14.4}{\familydefault}{\mddefault}{\updefault}{\color[rgb]{0,0,0}${\huge  k^n}$}%
}}}}
\put(4501,-5161){\makebox(0,0)[lb]{\smash{{\SetFigFont{12}{14.4}{\familydefault}{\mddefault}{\updefault}{\color[rgb]{0,0,0}${\huge \mathbf{x}}$}%
}}}}
\put(4651,-1936){\makebox(0,0)[lb]{\smash{{\SetFigFont{12}{14.4}{\familydefault}{\mddefault}{\updefault}{\color[rgb]{0,0,0}${\huge L_{m+n,\mathbf{x}}}$}%
}}}}
\end{picture}%

%% file: join.pstex_t
\begin{picture}(0,0)%
\includegraphics{join.pstex}%
\end{picture}%
\setlength{\unitlength}{3947sp}%
\begingroup\makeatletter\ifx\SetFigFont\undefined%
\gdef\SetFigFont#1#2#3#4#5{%
  \reset@font\fontsize{#1}{#2pt}%
  \fontfamily{#3}\fontseries{#4}\fontshape{#5}%
  \selectfont}%
\fi\endgroup%
\begin{picture}(4533,3366)(1801,-4594)
\put(1801,-2761){\makebox(0,0)[lb]{\smash{{\SetFigFont{12}{14.4}{\familydefault}{\mddefault}{\updefault}{\color[rgb]{0,0,0}$X$}%
}}}}
\put(5401,-3211){\makebox(0,0)[lb]{\smash{{\SetFigFont{12}{14.4}{\familydefault}{\mddefault}{\updefault}{\color[rgb]{0,0,0}$Y$}%
}}}}
\put(3676,-2161){\makebox(0,0)[lb]{\smash{{\SetFigFont{12}{14.4}{\familydefault}{\mddefault}{\updefault}{\color[rgb]{0,0,0}$J(X,Y)$}%
}}}}
\end{picture}%

%% file: thickened.pstex_t
\begin{picture}(0,0)%
\includegraphics{thickened.pstex}%
\end{picture}%
\setlength{\unitlength}{3947sp}%
\begingroup\makeatletter\ifx\SetFigFont\undefined%
\gdef\SetFigFont#1#2#3#4#5{%
  \reset@font\fontsize{#1}{#2pt}%
  \fontfamily{#3}\fontseries{#4}\fontshape{#5}%
  \selectfont}%
\fi\endgroup%
\begin{picture}(5951,5095)(1876,-5519)
\put(7351,-5461){\makebox(0,0)[lb]{\smash{{\SetFigFont{12}{14.4}{\familydefault}{\mddefault}{\updefault}{\color[rgb]{0,0,0}$T_0$}%
}}}}
\put(1876,-661){\makebox(0,0)[lb]{\smash{{\SetFigFont{12}{14.4}{\familydefault}{\mddefault}{\updefault}{\color[rgb]{0,0,0}$T_1$}%
}}}}
\put(5101,-2311){\makebox(0,0)[lb]{\smash{{\SetFigFont{12}{14.4}{\familydefault}{\mddefault}{\updefault}{\color[rgb]{0,0,0}$\Delta^1$}%
}}}}
\put(5701,-3511){\makebox(0,0)[lb]{\smash{{\SetFigFont{12}{14.4}{\familydefault}{\mddefault}{\updefault}{\color[rgb]{0,0,0}Thickened $1$-simplex}%
}}}}
\put(3826,-1561){\makebox(0,0)[lb]{\smash{{\SetFigFont{12}{14.4}{\familydefault}{\mddefault}{\updefault}{\color[rgb]{0,0,0}Thickened $0$-simplices}%
}}}}
\end{picture}%